\documentclass[conference]{IEEEtran}
\IEEEoverridecommandlockouts
\usepackage{cite}
\usepackage{amsmath,amssymb,amsfonts}
\usepackage{algorithmic}
\usepackage{graphicx}
\usepackage{textcomp}
\usepackage{xcolor}
\def\BibTeX{{\rm B\kern-.05em{\sc i\kern-.025em b}\kern-.08em
    T\kern-.1667em\lower.7ex\hbox{E}\kern-.125emX}}

\usepackage{mathrsfs}
\usepackage{amsthm}
\usepackage{algorithm}
\begin{document}
	
\title{The Error Probability of Spatially Coupled Sparse Regression Codes over Memoryless Channels
	}
\author{%
    \IEEEauthorblockN{YuHao Liu\IEEEauthorrefmark{1},
         Yizhou Xu\IEEEauthorrefmark{4}
         and TianQi Hou\IEEEauthorrefmark{3}
    }    
    \IEEEauthorblockA{\IEEEauthorrefmark{1}%
        Department of Mathematical Sciences, Tsinghua University, Beijing, China}\
    \IEEEauthorblockA{\IEEEauthorrefmark{4}Infplane AI Technologies Ltd, Beijing, China}
    \IEEEauthorblockA{\IEEEauthorrefmark{3}%
        Theory Lab, Central Research Institute, 2012 Labs, Huawei Technologies Co., Ltd.
        }
    \IEEEauthorblockA{%
        Emails: yh-liu21@mails.tsinghua.edu.cn,
        yizhou.xu.cs@gmail.com,
        thou@connect.ust.hk
    }
}

\maketitle
\footnotetext[1]{Yuhao Liu and Yizhou Xu contributed equally to this work.}

\begin{abstract}
Sparse Regression Codes (SPARCs) are capacity-achieving codes introduced for communication over the Additive White Gaussian Noise (AWGN) channels and were later extended to general memoryless channels. In particular it was shown via \emph{threshold saturation} that Spatially Coupled Sparse Regression Codes (SC-SPARCs) are capacity-achieving over general memoryless channels \cite{barbier2019universal} when using an Approximate Message Passing decoder (AMP). This paper, for the first time rigorously, analyzes the non-asymptotic performance of the Generalized Approximate Message Passing (GAMP) decoder of SC-SPARCs over memoryless channels, and proves exponential decaying error probability with respect to the code length.
\end{abstract}
	
\begin{IEEEkeywords}
Sparse Regression Codes, Spatial Coupling, capacity-achieving, Generalized Approximate Message Passing
\end{IEEEkeywords}
	
\section{Introduction}
Sparse Regression Codes (SPARCs), also called Sparse Superposition codes, were proposed by Joseph and Barron in \cite{joseph2012least} as a computationally-effective alternative to current coded modulation schemes. In contrast to other coding schemes\cite{gallager1962low,berrou1996near,arikan2009channel} that are provably capacity-achieving only over discrete channel, SPARCs have demonstrated the ability to achieve the Shannon capacity over Additive White Gaussian Noise (AWGN) channels \cite{joseph2012least}.

In \cite{joseph2012least}, SPARCs are shown to be capacity-achieving with the naive maximum likelihood decoder. However, the exponential computational complexity of the maximum likelihood decoder renders it impractical for implementation in real systems. To address this issue, two practical decoders with quardratic computational were then proposed: the adaptive successive hard-decision decoder\cite{joseph2013fast} and the iterative soft-decision decoder\cite{barron2012high}. While both decoders have been shown to asymptotically reach the Shannon capacity, reproducing these asymptotic results for any reasonable finite block length is challenging. This has led to the exploration of decoders based on Approximate Message Passing (AMP).

Due to the sparsity inherent in SPARCs, the reconstruction of the original message can be formulated as a compressed sensing problem. Subsequently, the efficient AMP algorithm\cite{donoho2010message,bayati2011dynamics}, originally developed for compressed sensing problems, was adapted for SPARCs decoders \cite{barbier2017approximate}. The AMP-based decoder, also characterized by quadratic complexity, showed significantly improved performance for finite block lengths. Rangan et al. introduced the Generalized Approximate Message Passing (GAMP) algorithm\cite{rangan2011generalized}, specifically designed for estimating random vectors subject to generic, component-wise, probabilistic channels. Following this, Barbier et al. formulated decoders for SPARCs over memoryless channels that exploit the capabilities of the GAMP algorithm\cite{biyik2017generalized}. However, the decoder based on (G)AMP can only achieve successful decoding when the communication rate $R$ is below the algorithmic threshold \cite{biyik2017generalized,barbier2017approximate}. To bridge the statistical-computational gap, structured SPARCs are constructed using two strategies: power allocation and spatial coupling.

Power Allocation (PA), proposed by Barron and Joseph \cite{joseph2012least}, aims to achieve the Shannon capacity with homogeneous coding matrices and has proven to be beneficial for the AMP decoder. Rush et al. have rigorously proven that SPARCs with exponential power allocation, equipped with the AMP decoder, can achieve exponentially decaying error probability over the AWGN channel\cite{rush2017capacity,rush2018error}. Similar analyses, employing either statistical physics or the rigorous conditional technique, have been extended to rotational-invariant coding scheme for complexity reduction\cite{hou2022sparse,liu2022sparse,xu2023capacity}. They demonstrated that SPARCs with appropriate PA and the Vector Approximate Message Passing (VAMP) decoder are also capacity-achieving  over the AWGN channel, provided the coding scheme satisfies spectrum criterion. However, PA empirically performs worse than SC, and up till now, there is no proof that PA could help SPARCs approach the Shannon capacity over generic channels. Therefore, we focus on spatial coupling in this paper.

Spatial Coupling (SC) was originally proposed for Low-Density Parity-Check (LDPC) codes \cite{felstrom1999time} as a way to practically achieve capacity. It has found successful applications in various information-theoretic problems, including error correcting code, constraint satisfaction and compressed sensing\cite{kudekar2011threshold,hamed2013threshold,krzakala2012statistical,donoho2013information}, to boost the performance for iterative algorithms. Barbier et al. subsequently proposed Spatially Coupled SPARCs (SC-SPARCs) for transmission over AWGN channel \cite{barbier2015approximate,biyik2017generalized,barbier2017approximate}. SC-SPARCs was observed to have enhanced robustness and improved reconstruction empirically compared with PA \cite{barbier2015approximate}. \cite{rush2021capacity} rigorously derived the non-asymptotic error probability of SC-SPARCs over the AWGN channel and proved its capacity-achieving property. For generic channels, \cite{cobo2023bayes} rigorously proved that, for fully factorized prior, GAMP algorithm with SC sensing matrices was able to approach asymptotically the minimum Mean-squared error (MMSE) estimator. These encouraging results motivate us to work on the capacity-achieving property of SC-SPARCs over generic channels in a rigorous manner. While similar arguments have been made through the threshold saturation analysis of the state evolution (SE) obtained by a non-rigorous replica trick \cite{barbier2019universal}, we aim to provide an absolutely rigorous proof.

\textit{Main contributions:}
We consider SC-SPARCs over general memoryless channels with the following contributions:
\begin{itemize}
    \item We delineate the traveling wave mechanics in GAMP decoder for SC-SPARCs non-asymptotically in Lemma \ref{lemma:SE} for the first time.
    \item We derive the exponentially decaying error probability for GAMP decoder of SC-SPARCs in Theorem \ref{theo:main} for $R<C$ (thus capacity-achieving). This is the first rigorous proof for capacity-achieving of SC-SPARCs over generic memoryless channel. 
\end{itemize}

Our results differ from existing literature in following ways. While our analysis is primarily based on \cite{rush2021capacity}, its results for AWGN channel cannot be directly generalized to generic memoryless channel because their analysis relies on the specific properties of the AWGN channel. Further, the asymptotic results in \cite{cobo2023bayes} are not applicable to SC-SPARCs as the section-wise structure of SPARCs requires a stronger convergence. Finally, we did not apply the potential function-based approach in \cite{barbier2019universal} and \cite{cobo2023bayes} in this paper. This decision was made on the belief that the potential function may not be easily employed for non-asymptotic analysis.

\textit{Notation}: We use $[n]$ to represent the set $\{1,2, \cdots, n\}$ for a positive $n$. The sign $\sharp$ denotes the cardinality of a set. We use boldface for matrices and vectors, and plain font for scalars. $\mathcal{D} z$ denotes the Gaussian measure, that is, $\frac{1}{\sqrt{2 \pi}} e^{-\frac{1}{2} z^2} \mathrm{d} z$, and $A \setminus B$ represents $A \cap B^{c}$.

\section{Settings}

\subsection{Code Structure} \label{sec:code}
We consider communication over a generic memoryless channel, where the output symbol $y$ is generated given the input symbol $u$ as $y = h(u,\omega)$. Here $\omega$ represents an unknown noise and $h:\mathbb{R}^2 \to \mathbb{R}$ is a known function. Given the probability law of $\omega$, the channel can be equivalently characterized using the conditional distribution $P_{\text{out}} (y | u)$. For a codeword $\boldsymbol{x} = (x_1, x_2, \cdots, x_n)$ transmitted over $n$ uses of the channel with average power constraint $\sum_{i=1}^n |x_i|^2 / n =1$, the Shannon capacity of the channel is $\mathcal{C} = I(Y;Z)$, equal to the input-output mutual information, where $Z \sim \mathcal{N}(0,1)$ and $Y \sim P_{\text{out}}(\cdot | Z)$. This simple model covers many standard communication channels as follows:
\begin{itemize}
    \item \textbf{AWGN channel:} $y = u + \omega $ where $\omega \sim \mathcal{N}(0,\sigma^2)$ is a white Gaussian noise. The Shannon capacity is $\mathcal{C}_{\text{AWGN}} = \frac{1}{2} \log_2 (1+\sigma^2)$.
    \item \textbf{Binary Erasure Channel (BEC):} $y = (1-\omega) \text{sign} (u)$, where $\omega$ is a Bernoulli random variable with success probability $\epsilon$. 
     The Shannon capacity is $\mathcal{C}_{\text{BEC}} = 1 - \epsilon$. 
    \item \textbf{Binary Symmetric Channel (BSC):} 
     $y = (1-\omega) \text{sign} (u) - \omega \text{sign} (u)$, where $\omega$ is a Bernoulli random variable with success probability $\epsilon$. The Shannon capacity is $\mathcal{C}_{\text{BSC}} = 1-h_2(\epsilon)$, where $h_2(\epsilon)$ is the Shannon entropy of the random variable $\omega$.
\end{itemize}

SPARCs can be described by a linear transform $\boldsymbol{x}=\boldsymbol{A}\boldsymbol{\beta}$ of the codeword $\boldsymbol{\beta} \in \mathbb{R}^{N}$ with the design (or coding) matrix $\boldsymbol{A} \in \mathbb{R}^{n \times N}$. The codeword $\boldsymbol{\beta}$ consists of $L$ non-overlapping sections each with $M$ elements, and thus $N=ML$. There is only one non-zero element in each section of $\boldsymbol{\beta}$, and we set its value to one. $\boldsymbol{\beta}$ is uniformly distributed over all possible $M^L$ messages. The communication rate is therefore $R = \frac{1}{n}L\log M$, as each message $\boldsymbol{\beta}$ uniquely maps to one of the $\log(M^L)$-length strings of input bits.

\begin{figure}[htbp]
    \centering	
    \includegraphics[width=0.95\linewidth]{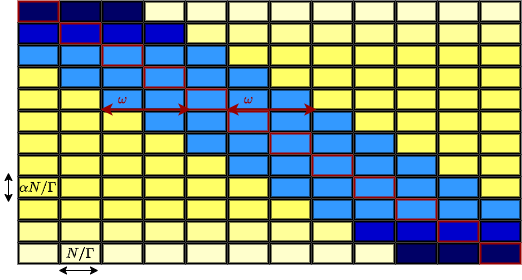}
    \vspace{-1em}
    \caption{A spatially coupled coding matrix for SC-SPARCs with parameters $\omega = 2, \Gamma = 12$. The matrix $\mathbf{A}$ is partitioned to $\Gamma \times \Gamma$ blocks, indexed by $(r,c)$ and labeled as $\mathbf{A}_{rc}$. Each of these block comprises $N / \Gamma $ columns and $n/\Gamma = \alpha N/\Gamma$ rows where $\alpha = (\log M) / MR$. The i.i.d elements in $\mathbf{A}_{rc}$ are distributed as $\mathcal{N}(0, W_{rc} / L)$. In each row, variances in blue blocks collectively share the power $(1-\rho)\Gamma$ equally, while the variances in yellow blocks share the power $\rho \Gamma$. This meticulous design satisfies the row normalization condition $\frac{1}{\Gamma} \sum_{c} W_{rc} = 1, \forall r$. The design matrix degrades to a special case of SC ensembles in \cite{barbier2019universal} when $\rho=0$.
    \label{fig:fig1}
    }   
\end{figure}
    
In SC-SPARCs, the design matrix $\boldsymbol{A}$ is partitioned into $\Gamma \times \Gamma$ blocks indexed by $(r,c)$, each with $N/\Gamma$ columns and $n/ \Gamma$ rows. The rows and columns in the block $(r,c)$ are collected by the sets $\mathcal{I}_r$ and $\mathcal{J}_c$, respectively, denoted by
\begin{equation}
\begin{gathered}
\mathcal{I}_r = \{(r-1)\frac{n}{\Gamma}+1, \cdots, r \frac{n}{\Gamma} \}, \quad \text{for} \; r \in [\Gamma] \\
\mathcal{J}_c = \{(c-1)\frac{N}{\Gamma}+1, \cdots, c \frac{N}{\Gamma} \}, \quad \text{for} \; c \in [\Gamma] 
\end{gathered}
\end{equation}
The block $(r,c)$ is then represented by $\boldsymbol{A}_{rc}$ containing entries $A_{ij}$ for $i \in \mathcal{I}_r$ and $j \in \mathcal{J}_c$. 
The matrix $\boldsymbol{A}$ consists of independently Gaussian distributed entries with variances specified by the base matrix $\boldsymbol{W} \in \mathbb{R}^{\Gamma \times \Gamma}$.  Entries in $\mathbf{A}_{rc}$ are i.i.d distributed as $\mathcal{N}(0,\frac{W_{rc}}{L})$. In order to satisfy the average power constraint of $\boldsymbol{x}$, the base matrix should satisfy the variance normalization condition:
\begin{equation}
    \frac{1}{\Gamma^2} \sum_{r=1}^{\Gamma} \sum_{c=1}^{\Gamma} W_{rc} = 1.
\end{equation}

In this paper, we adopt the row variance normalization condition \cite{barbier2019universal} as $\frac{1}{\Gamma} \sum_{c=1}^{\Gamma} W_{rc} = 1$ for all $r \in [\Gamma]$.
This normalization enforces homogeneous power over all components of $\boldsymbol{x}$. The construction of the base matrix with variance normalization can be achieved through various methods. We utilize the following base matrix inspired by SC-LDPC codes \cite{mitchell2015spatially}
\begin{equation}
    W_{r c}=\left\{\begin{array}{ll}
    (1-\rho) \cdot \frac{\Gamma}{\gamma_r} & \text { if } |r-c| \leq \omega \\
    \rho  \cdot \frac{\Gamma}{\Gamma - \gamma_r} & \text { otherwise. }
    \end{array}\right.
    \vspace{-0.3em}
\end{equation}
where $\gamma_r = \sharp \{c| 1 \leq c \leq \Gamma , |r-c| \leq \omega \}$. The SC-SPARCs design matrix is then parameterized by the tuple $(n, R, M, \Gamma, \omega, \rho)$. An example with parameters $\omega = 2, \Gamma = 12$ is shown in Fig.~\ref{fig:fig1}. We note that $\rho=0$ is preferred for practical implementation, but we need to choose a small positive $\rho$ proportional to the rate gap from capacity to address technical difficulties, as discussed in \cite{rush2021capacity}.

We also introduce a seed at the boundaries, assuming that the first $4\omega$ and last $4\omega$ sections of the transmitted codeword are known at the decoder side. The indexes of the seed are collected in the set $\mathcal{I} = \{1,\cdots, 4\omega \} \cup \{\Gamma-4\omega+1,\cdots, \Gamma\}$. The seed can be interpreted as perfect side information to enhance the performance through a 'reconstruction wave' propagating from boundary to middle\cite{barbier2019universal,rush2021capacity}, under (G)AMP decoding. The seeds will result in a rate loss so that the effective rate reads
	\begin{equation}
		R_{\text{eff}}=R(1-\frac{8\omega}{\Gamma}).
	\end{equation}
However, by setting $\Gamma>\omega^2$, the rate loss becomes negligible when $\omega\to\infty$ at any fixed $R$.

\subsection{GAMP decoder for SC-SPARCs}
\begin{algorithm}[h]
    \caption{GAMP decoder for SC-SPARCs}
    \begin{algorithmic}[1]
    \REQUIRE Max iteration T, design matrix $\boldsymbol{A}\in\mathbb{R}^{n\times N}$, observation $\boldsymbol{y}\in\mathbb{R}^n$.
    \STATE Initialization: $\boldsymbol{\beta}^{-1}_c=\boldsymbol{\beta}_c$ for $c \in \mathcal{I} $ , and $\boldsymbol{\beta}_c^{-1}=0$ otherwise . $\boldsymbol{s}^{-1}_r=0$ for $r\in[\Gamma]$
    \FOR { t = 0 to T }   
		\STATE  $\boldsymbol{p}_r^t=\sum_{c=0}^\Gamma\boldsymbol{A}_{rc}\boldsymbol{\beta}^{t-1}_c-\sigma^t_r\boldsymbol{s}^{t-1}_r$, $r\in[\Gamma]$
		\STATE $\boldsymbol{s}^t_r=\mathbf{g}_{\text{out}}(\boldsymbol{p}^t_r,\boldsymbol{y}_r,\sigma^t_r)$, $r\in[\Gamma]$
		\STATE $\boldsymbol{r}^t_c=\boldsymbol{\beta}^t_c+\tau^t_c\sum_{r=1}^R\boldsymbol{A}_{rc}^T\boldsymbol{s}^t_r$, $c\in[\Gamma] \setminus \mathcal{I} $
		\STATE $\boldsymbol{\beta}^t_c=\mathbf{g}_{\text{in}}(\boldsymbol{r}^t_c,\tau^t_c)$, $c\in[\Gamma] \setminus \mathcal{I}$
        \STATE Keep $\boldsymbol{\beta}_c^t = \boldsymbol{\beta}_c$ for $c \in \mathcal{I}$ throughout the iteration
    \ENDFOR
    \RETURN $\boldsymbol{\hat{\beta}}$ obtained from $\boldsymbol{\beta}^T$ by setting the largest entry in each section to $1$ and others to $0$
\end{algorithmic}
\label{algo:GAMP}
\end{algorithm}

\vspace{-0.5em}
The GAMP decoder for SC-SPARCs is presented in Algorithm \ref{algo:GAMP}. Before we proceed, we introduce the notation: when a row- or column- block index ($r$ or $c$) appears as a a subscript of a vector it denotes the corresponding block of that vector containing elements with indexes belong to $\mathcal{I}_r$ or $\mathcal{J}_c$ (e.g., $\boldsymbol{\beta}_c = ({\beta}_i)_{i \in \mathcal{J}_c}$).

The crux of the algorithm lies in the iterative resolution of two decoupled channel estimation problems. The first estimator, denoted as $\mathbf{g}_{\text{in}}(\boldsymbol{r}, \tau)$, is related to SPARCs with the expression as 
\vspace{-0.8em}
\begin{equation}
    [\mathbf{g}_{\text{in}}(\boldsymbol{r},\tau)]_i=\frac{e^{r_i/\tau}}{\sum_{j\in\text{sec}(\ell_i)}e^{r_j/\tau}}
\end{equation}
\newpage
where $\ell_i$ signifies section to which the $i$-th scalar component belongs and the operator $\text{sec}(\ell):=\{j\in\mathbb{N}\;|\;M(\ell-1)<j\le M\ell\}$ gathers the indices of elements belonging to the $\ell$-th section. The estimator $\mathbf{g}_{\text{in}}$ can be interpreted as the minimum mean-square error (MMSE) estimator for the channel  
\begin{equation}
    \boldsymbol{R}= \boldsymbol{\beta}_{0} + \sqrt{\tau} \boldsymbol{Z}
\end{equation}
where $\boldsymbol{Z} \sim \mathcal{N} (\boldsymbol{0}, \boldsymbol{I}_{N/\Gamma})$ and $\boldsymbol{\beta}_0 \sim P_0 ^ {\otimes L/\Gamma}$ independent of $\boldsymbol{Z}$. $P_0$ represents  the uniform distribution over the $M$- dimensional vectors with a single non-zero entry equal to $1$, i.e., $ P_0(\boldsymbol{x}) = M^{-1} \sum_{i\in [M]} \delta_{x_i,1} \prod_{j \neq i} \delta_{x_j, 0}$.
The estimator $\textbf{g}_{\text{out}}$ is related to the memoryless channel with the following expression:
\begin{equation}
    \mathbf{g}_{\text{out}} (\boldsymbol{p}, \boldsymbol{y}, \sigma) = \frac{\mathbb{E}_p \boldsymbol{z} - \boldsymbol{p}}{\sigma}
\end{equation}
where the expectation is taken with respect to the posterior distribution  
$p(\boldsymbol{z} | \boldsymbol{y}) \propto P_{\text{out}} (\boldsymbol{y} | \boldsymbol{z}) \exp{(-\frac{1}{2 \sigma} \|\boldsymbol{z} - \boldsymbol{p}\|_2^2)}$.
The term $\mathbb{E}_p \boldsymbol{z}$ in the estimator $\mathbf{g}_{\text{out}}$ can be interpreted as the MMSE estimator over the channel given the observation $\mathbf{Y} = \boldsymbol{y}$
\begin{equation}
    \mathbf{Z} \sim \mathcal{N}(\boldsymbol{}{p}, \sigma \mathbf{I}_{n/\Gamma}), \quad \mathbf{Y} \sim P_{\text{out}} (\cdot | \mathbf{Z}).
\end{equation}
we refer to \cite{rangan2011generalized} for additional background on GAMP.

The GAMP decoder for SPARCS in Algorithm \ref{algo:GAMP} was first introduced in \cite{biyik2017generalized} together with SE to track its asymptotic large-system behavior. Recall that SE is a set of recursions used to predict the performance of the algorithm, see \cite{bayati2011dynamics}. The SE in the present context reads
\begin{equation} 
    \label{SE}
    \begin{aligned}  
		&\sigma^t_r=\frac{1}{\Gamma}\sum_{c=1}^\Gamma W_{rc}\psi_c^t,\quad \phi^t_r=f_{\text{out}}(\sigma^t_r)^{-1},\quad r\in[\Gamma]\\
		&\tau^t_c=\frac{R}{\ln M}\left[\frac{1}{\Gamma}\sum_{r=1}^\Gamma\frac{W_{rc}}{\phi^t_r}\right]^{-1},\quad\psi^{t+1}_c=1-\varepsilon(\tau^t_c),\quad c\in[\hat{\Gamma}] ,
    \end{aligned}
\end{equation}
with the initialization $\psi_c^0 = 0 $ (keep $\psi_c^t =0$) for $c \in \mathcal{I}$ and $\psi_c^0 = 1 $ for $c \in [\hat{\Gamma}]$  where 
$[\hat{\Gamma}] = [\Gamma] \setminus \mathcal{I}$. In the following, we assume that row and column indexes $r$ and $c$ within the range $r \in [\Gamma], c\in[\hat{\Gamma}]$. The expression of $f_{\text{out}}(\sigma)$ is
\begin{equation} \label{eq:fout1}
    f_{\text{out}}(\sigma_r^t) = -\mathbb{E} \frac{\partial}{\partial p} g_{\text{out}} (Y,P,\sigma^t_r),
\end{equation}
and the expectation is taken over $(Y,P)$, where $Y \sim P_{\text{out}} (\cdot | Z_0)$ and $(P,Z_0)$ follows a joint Gaussian distribution with $\mathbb{E} [Z_0^2] = 1$, $\mathbb{E} [P^2] = \mathbb{E} [PZ_0] = 1 - \sigma^t_r$.
The expression of $\epsilon(\tau_c^t)$ is
\begin{equation}
    \varepsilon(\tau^t_c)=\mathbb{E}\left[\frac{e^{U_1/\sqrt{\tau^t_c}}}
    {e^{U_1/\sqrt{\tau^t_c}}+e^{-1/\tau_c^t\sum_{j=2}^Me^{U_j/\sqrt{\tau^t_c}}}}\right].
\end{equation}

The SE recursion $\psi^t_c\approx\frac{||\boldsymbol{\beta}_c^t-\boldsymbol{\beta}_c||_2^2}{L/\Gamma}$ can closely track the decoding mean square error (MSE). However, we need to stress that it is never rigorously proved non-asymptotically, which is the scope of this paper.

\section{Decoding Progression According to SE} 
\label{sec:decrease_SE}
Before illustrating the decreasing progress of SE, we elucidate three properties of $f_{\text{out}}$ that prove useful in the subsequent analysis.
\newtheorem{proposition}{Proposition}
\begin{proposition} 
\label{prop:prop1}
   For a continuously differentiable conditional probability density function $P_{\text{out}} (y | x)$ with respect to $x$,
   the expression of $f_{\text{out}}(\sigma)$ in (\ref{eq:fout1}) is
   \begin{equation} \label{eq:fout2}
		f_{\text{out}}(\sigma)=\int \mathrm{d}y \mathcal{D}\xi\frac{[P_{\text{out}}^{\prime}(y|\sqrt{\sigma}z+\sqrt{1-\sigma}\xi) \mathcal{D} z]^2}{\int P_{\text{out}}(y|\sqrt{\sigma}z+\sqrt{1-\sigma}\xi) \mathcal{D} z}
   \end{equation}
   where $P^{\prime}_{\text{out}} (y|x) = \frac{\partial}{\partial x} P_{\text{out}} (y|x)$. Then,  $f_{\text{out}} (\sigma)$ is non-negative for all $\sigma \in [0,1]$.
\end{proposition}

\begin{proposition}
    \label{prop:prop2}
    $f_{\text{out}}(\sigma) = -2 \frac{\mathrm{d} }{\mathrm{d} \sigma} \Psi_{\text{out}}(\sigma)$ where $\Psi_{\text{out}}(\sigma)$
   is a potential function defined by
    \begin{equation}
		\Psi_{\text{out}}(\sigma)=\mathbb{E}\log\int \mathcal{D}z P_{\text{out}}(Y|\sqrt{\sigma}z+\sqrt{1-\sigma}\xi),
    \end{equation} 
    Here, the expectation is taken with respect to $Y\sim P_{\text{out}}(\cdot|\sqrt{\sigma}z^*+\sqrt{1-\sigma}\xi)$, where $z^*\sim\mathcal{N}(0,1)$, $\xi\sim\mathcal{N}(0,1)$, and $z^*$, $\xi$ are independent. Furthermore, the Shannon capacity can be expressed in terms of $f_{\text{out}}(\sigma)$ as 
    \begin{equation} \label{eq:Shannon2}
       \mathcal{C}=\Psi_{\text{out}}(0) -\Psi_{\text{out}}(1) = \frac{1}{2}\int_0^1f_{\text{out}}(x) \mathrm{d}x. 
    \end{equation}   
\end{proposition}

\begin{proposition}
   \label{prop:prop3}
   The function $f_{\text{out}}(\sigma)$ is non-increasing for all $\sigma \in [0,1]$.
\end{proposition}
The proofs of Propositions \ref{prop:prop1} and \ref{prop:prop2} are similar to the methodology shown in \cite[Claim 1]{barbier2019universal}, which are given in Appendix \ref{app:proof}. Proposition \ref{prop:prop3} is a direct corollary from \cite[Proposition 20]{barbier2019optimal}. 

Equipped with above propositions, we can delineate the decoder progress of SE with the following lemma:
\newtheorem{lemma}{Lemma}
\begin{lemma}
    We consider $M$ large enough. Define $\Delta=\mathcal{C}-R$ and $\delta\in(0,\min\{\frac{\Delta}{2R},\frac{1}{2}\})$. Assume $0<\rho< h(\Delta)/2$, where $h(\Delta)>0$ is the solution of
    \begin{equation}
		\int_0^{h} f_{out}(x)dx=\frac{3}{2}\Delta.
		\label{eq:h(Delta)}
    \end{equation}
    If the rate satisfies $R<\frac{1-\rho}{2+\delta}f_{out}(1)$, all sections are simultaneously decoded at one iteration, i.e., for all $c\in[\Gamma]$,
    \begin{equation}
		\psi_c^1\leq f_{M,\delta}:=\frac{M^{-k\delta}}{\delta\sqrt{\log M}},
		\label{eq:all_decoded}
    \end{equation}
    where $k>0$ is a universal constant.
		
    Otherwise, if the rate satisfies $\frac{1-\rho}{2+\delta}f_{out}(1)\leq R< \mathcal{C}$. Let
    \begin{equation}
		\begin{aligned}
		&g=\max_k\{k\in\mathbb{N}|k\leq2\omega+1,\\
		&\qquad\int_0^{2\rho+\frac{k(1-\rho)}{2\omega+1}}f_{out}(x)dx-\frac{k}{2\omega+1}(1-\rho)f_{out}(1)<\frac{3}{2}\Delta\}.
		\end{aligned}
		\label{eq:g(Delta)}
    \end{equation}
    Assume $\omega$ large enough such that $g\geq1$. Then for $t\geq1$ and $c\leq\max\{4\omega+tg,\lceil\frac{\Gamma}{2}\rceil\}$, we have
    \begin{equation}
		\psi_c^t=\psi_{\Gamma-c+1}^t\leq f_{M,\delta}.
		\label{eq:decode_progression}
    \end{equation}
    \label{lemma:SE}
\end{lemma}
\begin{proof}
Noticing the symmetry $\sigma^t_r = \sigma^t_{\Gamma -r +1} $ and $\Psi^t_c = \Psi^t_{\Gamma - c + 1}$ for $1 \leq r,c \leq 
\left\lceil\frac{\Lambda}{2}\right\rceil$, we carry out the analysis for $1 \leq r,c \leq 
\left\lceil\frac{\Lambda}{2}\right\rceil$; the result for the other half then holds by symmetry. According to the initialization $\psi_c^0$, we have 
\begin{equation} \label{eq:sigma0r}
    \sigma_r^0=\left\{\begin{array}{ll}
    \frac{\Gamma - 8 \omega}{ \Gamma - \omega - \min \{r, \omega + 1\}} \rho &  1 \leq r \leq 3 \omega \\
    \overline{k} \frac{1-\rho}{2\omega+1} + \frac{\Gamma - 8\omega -\overline{k}}{\Gamma - 2 \omega -1} \rho & r = 3 \omega + k \leq  \left\lceil\frac{\Lambda}{2}\right\rceil
        \end{array}\right.
\end{equation}
where $\overline{k} = \min \{k, 2\omega + 1\}$.

Let $F_c^t:=\frac{R}{\tau_c^t\ln M}$.\cite[Lemma 4.1]{rush2021capacity} shows that $\psi_c^{t+1}\leq f_{M,\delta}$ once $F_c^t>(2 + \delta)R$. We now obtain a lower bound on $F_c^0$ for indices $4\omega \leq c \leq 6\omega + 1$. Using (\ref{SE}) and (\ref{eq:sigma0r}) we have 
\begin{equation} 
    \begin{aligned}
		F_{4\omega+k}^0&=\frac{1} {\Gamma}\sum_{r=1}^\Gamma W_{r,4\omega+k}f_{\text{out}}(\sigma_r^0)\\
		&\geq \frac{(1-\rho)}{2\omega+1}[\sum_{r=k}^{2\omega+1}f_{\text{out}}(\rho+r\frac{1-\rho}{2\omega+1}) + (k-1)f_{\text{out}} (1) ].
    \end{aligned}
\end{equation}
for $1\leq k\leq2\omega+1$. The inequality holds because we scale the coefficient of $\rho$ in (\ref{eq:sigma0r}) to $1$ and $f_{\text{out}}$ is non-negative and non-increasing on $[0,1]$ (Proposition \ref{prop:prop1} and \ref{prop:prop3}). Thus, if $R<\frac{1-\rho}{2+\delta}f_{out}(1)$, we have $F_{4\omega+k}^0\geq(1-\rho)f_{out}(1)>(2+\delta)R$, then \eqref{eq:all_decoded} holds. Otherwise,
\begin{equation}
    \begin{aligned}
        F_{4\omega+k}^0
		&\geq(1-\rho)[\int_{\frac{k(1-\rho)}{2\omega+1}}^1f_{\text{out}}(\rho+(1-\rho)x)dx+\frac{k f_{\text{out}}(1)}{2\omega+1}]\\
        &= \int_{\rho + \frac{(1-\rho)k}{2\omega+1}}^1 f_{\text{out}}(x) \mathrm{d} x + \frac{k(1-\rho)}{2 \omega + 1} f_{\text{out}}(1) \\
		&=2\mathcal{C}-\int_0^{\rho+\frac{k(1-\rho)} {2\omega+1}}f_{\text{out}}(x)dx+\frac{k(1-\rho)}{2\omega+1}f_{\text{out}}(1),
    \end{aligned}
    \label{eq:F^0,upper bound}
\end{equation}
where the inequality is from the fact that $f_{\text{out}}$ is non-increasing (Proposition \ref{prop:prop3}) and we use Proposition \ref{prop:prop2} in the last equality. Therefore, using \eqref{eq:g(Delta)}, we have $F_{4\omega+k}^0>(2+\delta)R$ for all $k\leq g$.
		
Next we consider subsequent iterations $t>1$. Assume inductively that
\begin{equation}
    \psi_c^{t-1}=\psi_{\Gamma-c+1}^{t-1}\leq f_{M,\delta}
\end{equation}
for $c\leq tg$. For simplicity, we use the shorthand $f=f_{M,\delta}$. By the SE recursion \eqref{SE}, we deduce
\begin{equation}
    \sigma_r^t\leq\left\{\begin{aligned}
    &f (1-\rho)+\rho,\ 1\leq r\leq 3\omega+tg\\
    &\frac{k(1-f)(1-\rho)}{2\omega+1}+f(1-\rho)+\rho,\ r= 3\omega+tg+k\\
    &1,\ 5\omega+tg + 1<r<\lceil\frac{\Gamma}{2}\rceil
    \end{aligned}\right.
\end{equation}
where $1\leq k\leq2\omega+1$. Here we assume $\frac{\Gamma}{2}>5\omega+tg$, and the other circumstance is similar. Then for $1\leq k\leq2\omega+1$,
\begin{equation}
    \begin{aligned}
		F_{4\omega+tg+k}^t&\geq\frac{k(1-\rho)}{2\omega+1}\sum_{r=k}^{2\omega+1}\\
		&\qquad f_{\text{out}}\left(\frac{r(1-\rho)}{2\omega+1}+\frac{2\omega+1-r}{2\omega+1}f(1-\rho)+\rho\right)\\
		&\geq\frac{k(1-\rho)}{2\omega+1}f_{\text{out}}(1)+\int_{2\rho+(1-\rho)\frac{k}{2\omega+1}}^1f_{\text{out}}(x)dx,
    \end{aligned}
\end{equation}
where in the second inequality, we use the monotonicity of $f_{\text{out}}$ (Proposition \ref{prop:prop3}) and $f_{M,\delta}<\rho$ for $M$ large enough. By \eqref{eq:g(Delta)}, the right side is lower bounded by $(2+\delta)R.$ As $\psi_c^t=\psi_{\Gamma-c+1}^t$ holds simply by symmetry, now we finish the proof of \eqref{eq:decode_progression}.
		
Lastly, we need to verify that the set defined in \eqref{eq:g(Delta)} contains at least $k=1$. We notice that the the left side of inequality \eqref{eq:g(Delta)} is an non-decreasing function with respect to $k$.  When $k=0$, the inequality in \eqref{eq:g(Delta)} becomes 
\begin{equation}
    \begin{aligned}
		\int_0^{2\rho} f_{\text{out}}(x)dx<\frac{3}{2}\Delta
    \end{aligned}
\end{equation}
The inequality holds according to the constraint of $\rho$ and \eqref{eq:h(Delta)}. Then we can choose a sufficiently large $\omega$ (not related to $M$) to ensure that $g\geq1$.
\end{proof}
	
Lemma \ref{lemma:SE} indicates that after one iteration, at least $g$ sections are successfully decoded with an error $f_{M,\delta}$ that approaches zero for $M$ large enough. Therefore, we can run the SC-GAMP decoder for at most
\begin{equation}
    T=\left\lceil\frac{\Gamma}{2g}\right\rceil
    \label{eq:T}
\end{equation}
steps and obtain a small error rate.

\noindent
\textbf{Remark 1.} We adopt the base matrix with row normalization, similar to \cite{barbier2019universal} but in contrast to the column normalization approach illustrated in \cite{rush2021capacity,cobo2023bayes}, with the aim of ensuring uniformity in the expressions of $f_{\text{out}}(\sigma_r^t)$ for all  $r \in [\Gamma]$. Over the AWGN channel, column-normalization leads to $f_{\text{out}} (\sigma_r^t) = \frac{1}{\sigma_r^t + \sigma^2}$, a uniform expression for all $r \in [\Gamma]$, which does not extend to most memoryless channels. \cite{cobo2023bayes} uses the potential analysis to overcome this difficulty, which cannot be applied in this paper.
	
\section{Non-asymptotic performance of the GAMP decoder for SC-SPARC}
\label{sec:non-asy}
To depict the non-asymptotic performance of SG-GAMP for SPARC, some technique assumptions are required as follows:
\newtheorem{assumption}{Assumption}
\begin{assumption} \label{assumption1}
    $\epsilon$ is a sub-Gaussian variable.
\end{assumption}
\vspace{-1em}
\begin{assumption}  \label{assumption2}
    The estimator $g_{\text{out}}(p,y,\sigma)$ is continuous differentiable and Lipschitz on variables $p,u,\epsilon$ for different $\sigma$, where $y = h(u,\epsilon)$.
\end{assumption}  
\vspace{-1em}
\begin{assumption}   \label{assumption3}
    $\Psi_{\text{out}} (\sigma)$ is  twice continuously differentiable
    and strongly convex i.e. $\Psi_{out}^{\prime \prime} (\sigma) \geq L > 0$ on $[0,1]$.
\end{assumption}
\newpage
\noindent
\textbf{Remark 2.} For some channels (e.g. BEC), the Lipschitz constant in Assumption $\ref{assumption2}$ is a continuous function of $\sigma$. However, our analysis in Lemma \ref{le: bound} shows the existence of a lower bound of $\sigma_r^t$ for $r \in [\Gamma], 0 \leq t \leq T$. Then we can treat $g_{\text{out}}(p, y, \sigma)$ as a Lipschitz continuous function for different $\sigma_r^t$ during the iteration.
We note that Assumption \ref{assumption3} is not very strong, as Proposition \ref{prop:prop3} has shown that $\Psi_{\text{out}}$ has a non-negative second derivative.  We can verify that commonly used channels in communication satisfy Assumption \ref{assumption3}, such as the AWGN channel, BEC and BSC.

Finally we present the main theorem, Theorem \ref{theo:main}, demonstrating that if $R<\mathcal{C}$, the error probability of SC-SPARCs decays exponentially with $n$, and thus capacity-schieving.

\newtheorem{theorem}{Theorem}
\begin{theorem}
\label{theo:main}
Let $\kappa,K,K^{\prime},\xi,\Xi$ be universal constants not depending on $n,M,L,\epsilon,t$. For $t\geq0$, we define 
\begin{equation}  \label{eq:pa1}
    \begin{aligned}
         &K_{t}=\Xi^{2 t}(t !)^{14}, \ \kappa_{t}=\frac{1}{\xi^{2 t}(t !)^{24}}, \\
         &K_{t}^{\prime}=\Xi(t+1)^{8} K_{t},\ \kappa_{t}^{\prime}=\frac{\kappa_{t}}{\xi(t+1)^{12}}.
    \end{aligned}
\end{equation}
Let $ \epsilon_{\text{sec}}:=\frac{1}{L}\sum_{\ell=1}^L\mathbf{1}\{\hat{\boldsymbol{\beta}}_{\sec(\ell)}\neq\boldsymbol{\beta}_{\sec(\ell)}\}$ be the section error rate. Under the assumptions in Section \ref{sec:non-asy} and Lemma \ref{lemma:SE}, for $\epsilon\in(0,1)$, fix rate $\mathcal{R}<C$ and let $M$ be large enough such that $f_{M,\delta}\leq\frac{\epsilon}{8}$. Then after the SC-GAMP decoder runs for $T$ (defined in \eqref{eq:T}) steps, the section error rate satisfies
    \begin{equation} \label{eq:main}
		\mathbb{P}(\epsilon_{\text{sec}}>\epsilon)\leq K_{T-1}\Gamma^{2T+1}\exp\{-\frac{\kappa_{T-1}n\epsilon^2}{(\log M)^{2T}(\Gamma / \omega)^{2T+1}} \}.
    \end{equation}
\end{theorem}

\textbf{Proof sketch}: The proof of Theorem \ref{theo:main} primarily consists of two parts. Firstly, after finite number of steps, the SE prediction decreases to a small value upper bounded by $f_{M,\delta}$ according to Lemma \ref{lemma:SE}. Secondly, we show that the the decoding process is accurately tracked by the SE with an exponentially decaying error with respect to the code length. The second part is achieved by employing the conditional technique \cite{bayati2011dynamics} to analyze the asymptotic distributions of $\boldsymbol{p}_r^t$ and $\boldsymbol{r}_c^t$ throughout iteration in Algorithm \ref{algo:GAMP}, based on the matrix-valued general recursion in \cite{rangan2011generalized} and \cite{cobo2023bayes}. The details are in Appendix \ref{app:concentration}. Moreover, Theorem \ref{theo:main} can be readily extended to spatially coupled general linear models (GLM), given in Appendix \ref{app:glm}.

 
 \section{Future Work}	
 Recently \cite{takeuchi2023orthogonal} has established the optimality of AMP for spatially-coupled right-orthogonally invariant matrices, which might extend to SC-SPARCs. Further, it will be of great importance to extend our results to spatially-coupled discrete cosine transform matrices following \cite{dudeja2022universality,dudeja2022spectral}, given their low complexity and favorable practical performance. Another avenue for research involves exploring techniques to enhance the finite-length performance of SC-SPARCs, drawing inspiration from the strategies employed in \cite{greig2017techniques} for power allocation. Lastly, we see potential in investigating the combination of spatial coupling with other approximate message passing algorithms\cite{venkataramanan2022estimation,xu2023approximate,liu2022memory,tian2022generalized} with low complexity.

\clearpage
\bibliographystyle{IEEEtran}
\bibliography{main}

\begin{thebibliography}{10}
\providecommand{\url}[1]{#1}
\csname url@samestyle\endcsname
\providecommand{\newblock}{\relax}
\providecommand{\bibinfo}[2]{#2}
\providecommand{\BIBentrySTDinterwordspacing}{\spaceskip=0pt\relax}
\providecommand{\BIBentryALTinterwordstretchfactor}{4}
\providecommand{\BIBentryALTinterwordspacing}{\spaceskip=\fontdimen2\font plus
\BIBentryALTinterwordstretchfactor\fontdimen3\font minus
  \fontdimen4\font\relax}
\providecommand{\BIBforeignlanguage}[2]{{%
\expandafter\ifx\csname l@#1\endcsname\relax
\typeout{** WARNING: IEEEtran.bst: No hyphenation pattern has been}%
\typeout{** loaded for the language `#1'. Using the pattern for}%
\typeout{** the default language instead.}%
\else
\language=\csname l@#1\endcsname
\fi
#2}}
\providecommand{\BIBdecl}{\relax}
\BIBdecl

\bibitem{barbier2019universal}
J.~Barbier, M.~Dia, and N.~Macris, ``Universal sparse superposition codes with
  spatial coupling and gamp decoding,'' \emph{IEEE Transactions on Information
  Theory}, vol.~65, no.~9, pp. 5618--5642, 2019.

\bibitem{joseph2012least}
A.~Joseph and A.~R. Barron, ``Least squares superposition codes of moderate
  dictionary size are reliable at rates up to capacity,'' \emph{IEEE
  Transactions on Information Theory}, vol.~58, no.~5, pp. 2541--2557, 2012.

\bibitem{gallager1962low}
R.~Gallager, ``Low-density parity-check codes,'' \emph{IRE Transactions on
  information theory}, vol.~8, no.~1, pp. 21--28, 1962.

\bibitem{berrou1996near}
C.~Berrou and A.~Glavieux, ``Near optimum error correcting coding and decoding:
  Turbo-codes,'' \emph{IEEE Transactions on communications}, vol.~44, no.~10,
  pp. 1261--1271, 1996.

\bibitem{arikan2009channel}
E.~Arikan, ``Channel polarization: A method for constructing capacity-achieving
  codes for symmetric binary-input memoryless channels,'' \emph{IEEE
  Transactions on information Theory}, vol.~55, no.~7, pp. 3051--3073, 2009.

\bibitem{joseph2013fast}
A.~Joseph and A.~R. Barron, ``Fast sparse superposition codes have near
  exponential error probability for $ r<{\cal c} $,'' \emph{IEEE transactions
  on information theory}, vol.~60, no.~2, pp. 919--942, 2013.

\bibitem{barron2012high}
A.~R. Barron and S.~Cho, ``High-rate sparse superposition codes with
  iteratively optimal estimates,'' in \emph{2012 IEEE International Symposium
  on Information Theory Proceedings}.\hskip 1em plus 0.5em minus 0.4em\relax
  IEEE, 2012, pp. 120--124.

\bibitem{donoho2010message}
D.~L. Donoho, A.~Maleki, and A.~Montanari, ``Message passing algorithms for
  compressed sensing: I. motivation and construction,'' in \emph{2010 IEEE
  information theory workshop on information theory (ITW 2010, Cairo)}.\hskip
  1em plus 0.5em minus 0.4em\relax IEEE, 2010, pp. 1--5.

\bibitem{bayati2011dynamics}
M.~Bayati and A.~Montanari, ``The dynamics of message passing on dense graphs,
  with applications to compressed sensing,'' \emph{IEEE Transactions on
  Information Theory}, vol.~57, no.~2, pp. 764--785, 2011.

\bibitem{barbier2017approximate}
J.~Barbier and F.~Krzakala, ``Approximate message-passing decoder and capacity
  achieving sparse superposition codes,'' \emph{IEEE Transactions on
  Information Theory}, vol.~63, no.~8, pp. 4894--4927, 2017.

\bibitem{rangan2011generalized}
S.~Rangan, ``Generalized approximate message passing for estimation with random
  linear mixing,'' in \emph{2011 IEEE International Symposium on Information
  Theory Proceedings}.\hskip 1em plus 0.5em minus 0.4em\relax IEEE, 2011, pp.
  2168--2172.

\bibitem{biyik2017generalized}
E.~Biyik, J.~Barbier, and M.~Dia, ``Generalized approximate message-passing
  decoder for universal sparse superposition codes,'' in \emph{2017 IEEE
  International Symposium on Information Theory (ISIT)}.\hskip 1em plus 0.5em
  minus 0.4em\relax IEEE, 2017, pp. 1593--1597.

\bibitem{rush2017capacity}
C.~Rush, A.~Greig, and R.~Venkataramanan, ``Capacity-achieving sparse
  superposition codes via approximate message passing decoding,'' \emph{IEEE
  Transactions on Information Theory}, vol.~63, no.~3, pp. 1476--1500, 2017.

\bibitem{rush2018error}
C.~Rush and R.~Venkataramanan, ``The error probability of sparse superposition
  codes with approximate message passing decoding,'' \emph{IEEE Transactions on
  Information Theory}, vol.~65, no.~5, pp. 3278--3303, 2018.

\bibitem{hou2022sparse}
T.~Hou, Y.~Liu, T.~Fu, and J.~Barbier, ``Sparse superposition codes under vamp
  decoding with generic rotational invariant coding matrices,'' \emph{arXiv
  preprint arXiv:2202.04541}, 2022.

\bibitem{liu2022sparse}
Y.~Liu, T.~Fu, J.~Barbier, and T.~Hou, ``Sparse superposition codes with
  rotational invariant coding matrices for memoryless channels,'' \emph{arXiv
  preprint arXiv:2205.08980}, 2022.

\bibitem{xu2023capacity}
Y.~Xu, Y.~Liu, S.~Liang, T.~Wu, B.~Bai, J.~Barbier, and T.~Hou,
  ``Capacity-achieving sparse regression codes via vector approximate message
  passing,'' \emph{arXiv preprint arXiv:2303.08406}, 2023.

\bibitem{felstrom1999time}
A.~J. Felstrom and K.~S. Zigangirov, ``Time-varying periodic convolutional
  codes with low-density parity-check matrix,'' \emph{IEEE Transactions on
  Information Theory}, vol.~45, no.~6, pp. 2181--2191, 1999.

\bibitem{kudekar2011threshold}
S.~Kudekar, T.~J. Richardson, and R.~L. Urbanke, ``Threshold saturation via
  spatial coupling: Why convolutional ldpc ensembles perform so well over the
  bec,'' \emph{IEEE Transactions on Information Theory}, vol.~57, no.~2, pp.
  803--834, 2011.

\bibitem{hamed2013threshold}
S.~Hamed~Hassani, N.~Macris, and R.~Urbanke, ``Threshold saturation in
  spatially coupled constraint satisfaction problems,'' \emph{Journal of
  Statistical Physics}, vol. 150, pp. 807--850, 2013.

\bibitem{krzakala2012statistical}
F.~Krzakala, M.~M{\'e}zard, F.~Sausset, Y.~Sun, and L.~Zdeborov{\'a},
  ``Statistical-physics-based reconstruction in compressed sensing,''
  \emph{Physical Review X}, vol.~2, no.~2, p. 021005, 2012.

\bibitem{donoho2013information}
D.~L. Donoho, A.~Javanmard, and A.~Montanari, ``Information-theoretically
  optimal compressed sensing via spatial coupling and approximate message
  passing,'' \emph{IEEE transactions on information theory}, vol.~59, no.~11,
  pp. 7434--7464, 2013.

\bibitem{barbier2015approximate}
J.~Barbier, C.~Sch{\"u}lke, and F.~Krzakala, ``Approximate message-passing with
  spatially coupled structured operators, with applications to compressed
  sensing and sparse superposition codes,'' \emph{Journal of Statistical
  Mechanics: Theory and Experiment}, vol. 2015, no.~5, p. P05013, 2015.

\bibitem{rush2021capacity}
C.~Rush, K.~Hsieh, and R.~Venkataramanan, ``Capacity-achieving spatially
  coupled sparse superposition codes with amp decoding,'' \emph{IEEE
  Transactions on Information Theory}, vol.~67, no.~7, pp. 4446--4484, 2021.

\bibitem{cobo2023bayes}
P.~P. Cobo, K.~Hsieh, and R.~Venkataramanan, ``Bayes-optimal estimation in
  generalized linear models via spatial coupling,'' in \emph{2023 IEEE
  International Symposium on Information Theory (ISIT)}.\hskip 1em plus 0.5em
  minus 0.4em\relax IEEE, 2023, pp. 773--778.

\bibitem{mitchell2015spatially}
D.~G. Mitchell, M.~Lentmaier, and D.~J. Costello, ``Spatially coupled ldpc
  codes constructed from protographs,'' \emph{IEEE Transactions on Information
  Theory}, vol.~61, no.~9, pp. 4866--4889, 2015.

\bibitem{barbier2019optimal}
J.~Barbier, F.~Krzakala, N.~Macris, L.~Miolane, and L.~Zdeborov{\'a}, ``Optimal
  errors and phase transitions in high-dimensional generalized linear models,''
  \emph{Proceedings of the National Academy of Sciences}, vol. 116, no.~12, pp.
  5451--5460, 2019.

\bibitem{takeuchi2023orthogonal}
K.~Takeuchi, ``Orthogonal approximate message-passing for spatially coupled
  linear models,'' \emph{IEEE Transactions on Information Theory}, 2023.

\bibitem{dudeja2022universality}
R.~Dudeja, Y.~M. Lu, and S.~Sen, ``Universality of approximate message passing
  with semi-random matrices,'' \emph{arXiv preprint arXiv:2204.04281}, 2022.

\bibitem{dudeja2022spectral}
R.~Dudeja, S.~Sen, and Y.~M. Lu, ``Spectral universality of regularized linear
  regression with nearly deterministic sensing matrices,'' \emph{arXiv preprint
  arXiv:2208.02753}, 2022.

\bibitem{greig2017techniques}
A.~Greig and R.~Venkataramanan, ``Techniques for improving the finite length
  performance of sparse superposition codes,'' \emph{IEEE Transactions on
  Communications}, vol.~66, no.~3, pp. 905--917, 2017.

\bibitem{venkataramanan2022estimation}
R.~Venkataramanan, K.~K{\"o}gler, and M.~Mondelli, ``Estimation in rotationally
  invariant generalized linear models via approximate message passing,'' in
  \emph{International Conference on Machine Learning}.\hskip 1em plus 0.5em
  minus 0.4em\relax PMLR, 2022, pp. 22\,120--22\,144.

\bibitem{xu2023approximate}
Y.~Xu, T.~Hou, S.~Liang, and M.~Mondelli, ``Approximate message passing for
  multi-layer estimation in rotationally invariant models,'' in \emph{2023 IEEE
  Information Theory Workshop (ITW)}.\hskip 1em plus 0.5em minus 0.4em\relax
  IEEE, 2023, pp. 294--298.

\bibitem{liu2022memory}
L.~Liu, S.~Huang, and B.~M. Kurkoski, ``Memory amp,'' \emph{IEEE Transactions
  on Information Theory}, vol.~68, no.~12, pp. 8015--8039, 2022.

\bibitem{tian2022generalized}
F.~Tian, L.~Liu, and X.~Chen, ``Generalized memory approximate message passing
  for generalized linear model,'' \emph{IEEE Transactions on Signal
  Processing}, vol.~70, pp. 6404--6418, 2022.

\end{thebibliography}

\clearpage
\appendix
\subsection{Proof of propositions} \label{app:proof}
This section aims to prove Proposition \ref{prop:prop1} and Proposition \ref{prop:prop2} given in Section \ref{sec:decrease_SE}. 
\subsubsection{Proof of Proposition \ref{prop:prop1}}
\begin{proof}
    With double expectation theorem, we can rewrite the expectation in (\ref{eq:fout1}) as:
    \begin{equation} \label{eq:dExp}
        -\mathbb{E}_{Y,P} \frac{\partial}{\partial p} g_{\text{out}} (Y,P,\sigma) = - \mathbb{E}_{P} \mathbb{E}_{Y|P}   \frac{\partial}{\partial p} g_{\text{out}} (Y,P,\sigma).
    \end{equation}
    
    The conditional distribution of $Z$ given $P$ is $Z|P \sim N(P, \sigma)$. Consequently, the distribution of $Y|P$ is 
    \begin{equation} \label{eq:cond_P1}
        Y|P \sim \int P_{\text{out}} (\cdot | \sqrt{\sigma} z + P) \mathcal{D} z.
    \end{equation}
    Then conditional expectation in (\ref{eq:dExp}) can be computed as:
    \begin{equation} \label{eq:cond_res}
        \mathbb{E}_{Y|P} \frac{\partial}{\partial p} g_{\text{out}} (Y,P,\sigma) = -\int \mathrm{d} y \frac{[\int P_{\text{out}}^{\prime} (y | \sqrt{\sigma} z + P)\mathcal{D}z]^2}{\int  P_{\text{out}}(y|\sqrt{\sigma}z + P ) \mathcal{D}z}.
    \end{equation}
    Taking (\ref{eq:cond_res}) into (\ref{eq:dExp}), we calculate the expression for $f_{\text{out}} (\sigma)$ as:
    \begin{equation}
        f_{\text{out}}(\sigma)=\int \mathrm{d}y \mathcal{D}\xi\frac{[P_{\text{out}}^{\prime}(y|\sqrt{\sigma}z+\sqrt{1-\sigma}\xi) \mathcal{D} z]^2}{\int P_{\text{out}}(y|\sqrt{\sigma}z+\sqrt{1-\sigma}\xi) \mathcal{D} z}.
    \end{equation}
    Obviously, $f_{\text{out}}(\sigma)$ is non-negative.
\end{proof}
\subsubsection{Proof of Proposition \ref{prop:prop2}}
Prior to presenting the proof, we introduce the shorthand notation $P_{\text{out}}(y| \sqrt{\sigma} z + \sqrt{1-\sigma} \xi)$ as $l_{y,\xi,\sigma} (z)$, and $l^{\prime}_{y,\xi,\sigma}(z)$, $l^{\prime \prime}_{y,\xi,\sigma} (z)$ denote its first and second derivative with respect to $z$, respectively.
\begin{proof}
    Take the derivative of the function of $\Psi_{\text{out}} (\sigma)$, we get:
    \begin{equation} \label{eq:dere_1}
        \begin{aligned}
            & \frac{\mathrm{d} \Psi_{\text{out}}(\sigma) }{\mathrm{d} \sigma} =-\frac{1}{2}\int \mathrm{d} y \mathcal{D} \xi \mathcal{D} z (-\frac{1}{\sqrt{1-\sigma}} \xi - \frac{1}{\sqrt{\sigma}}z) l_{y,\xi,\sigma}^{\prime} (z) \\
            &  + \frac{1}{2\sqrt{\sigma}} \int \mathrm{d} y \mathcal{D} \xi [\int  z l_{y,\xi,\sigma}^{\prime}(z) \mathcal{D}z] \log [\int \mathcal{D}z l_{y,\xi,\sigma}(z) ] \\
            & - \frac{1}{2 \sqrt{1-\sigma}}\int \mathrm{d} y \mathcal{D} \xi  [\xi \int   l_{y,\xi,\sigma}^{\prime}(z) \mathcal{D}z] \log [\int \mathcal{D}z l_{y,\xi,\sigma}(z) ]
        \end{aligned}
    \end{equation}

The first term in (\ref{eq:dere_1}) is $0$ because $\int P_{\text{out}} (y|z) dy = 1$. With Stein's lemma, we can rewrite the second term in (\ref{eq:dere_1}) 
as:
\begin{equation} \label{eq:term2}
    \frac{1}{2} \int \mathrm{d} y \mathcal{D} \xi [\int l_{y,\xi,\sigma}^{\prime \prime} (z) \mathcal{D} z] \log [\int l_{y,\xi,\sigma} (z) \mathcal{D} z ],
\end{equation}
Similarly, the third term in (\ref{eq:dere_1}) can be rewritten as:
\begin{equation} 
    \begin{aligned}
        & -\frac{1}{2} \int \mathrm{d} y \mathcal{D} \xi [\int l_{y,\xi,\sigma}^{\prime \prime} (z) \mathcal{D} z] \log [\int l_{y,\xi,\sigma} (z) \mathcal{D} z ] \\
        & -\frac{1}{2} \int \mathrm{d} y \mathcal{D} \xi \frac{[\int  l_{y,\xi,\sigma}^{\prime \prime} \mathcal{D}z]^2}{\int l_{y,\xi,\sigma} (z) \mathcal{D} z},
    \end{aligned}
    \label{eq:term3}
\end{equation}

Taking (\ref{eq:term2}) and (\ref{eq:term3}) into  (\ref{eq:dere_1}), we get 
\begin{equation}
    \frac{\mathrm{d} \Psi_{\text{out}}(\sigma) }{\mathrm{d} \sigma} = \frac{1}{2} \int \mathrm{d}y \mathcal{D}\xi\frac{[P_{\text{out}}^{\prime}(y|\sqrt{\sigma}z+\sqrt{1-\sigma}\xi) \mathcal{D} z]^2}{\int P_{\text{out}}(y|\sqrt{\sigma}z+\sqrt{1-\sigma}\xi) \mathcal{D} z}.
\end{equation}
This implies $f_{\text{out}}(\sigma) = -2 \frac{\mathrm{d} \Psi_{\text{out}}(\sigma) }{\mathrm{d} \sigma}$.

Next, we calculate the mutual information random variables $Y \sim P_{\text{out}}(\cdot|Z), Z \sim \mathcal{N}(0,1)$. The Shannon entropy of $Y$ is:
\begin{equation} \label{eq:entropy}
   H(Y) = \int \mathrm{d} y (\int P_{\text{out}} (y|z) \mathcal{D}z) \log (\int P_{\text{out}} (y|z) \mathcal{D}z),
\end{equation}
which is the expression of $\psi_{\text{out}}(0)$. The conditional entropy of $Y$ given $Z$ is:
\begin{equation}
    H(Y|Z) = \int \mathcal{D}z \int \mathrm{d} y  P_{\text{out}}(y|z) \log P_{\text{out}}(y|z),
\end{equation}
which is the expression of $\Psi_{\text{out}}(1)$. Thus, the mutual information of $Y$ and $Z$ is
\begin{equation}
    I(Y;Z) = H(Y)-H(Y|Z) = \Psi_{\text{out}}(0) - \Psi_{\text{out}}(1).
\end{equation}
This is the Shannon capacity of the channel given in Section \ref{sec:code}.
\end{proof}

\subsection{Concentration on the SE}
\label{app:concentration}
This section aims to give Lemma \ref{lemma:concentration}, which shows that the SC-GAMP decoder can be accurately tracked by its SE, with exponentially fast concentration.
	
With the definition
\begin{equation} \label{eq:general}
    \begin{aligned}
        &\boldsymbol{b}^t_r:=\boldsymbol{p}^t_r-\boldsymbol{z}_{0r},\ \breve{\boldsymbol{m}}^t_r:=-\phi^t_r\boldsymbol{s}^t_r, \\
        &\breve{\boldsymbol{q}}^t_c:=\boldsymbol{\beta}^t_c-\boldsymbol{\beta}_c,\ \boldsymbol{h}^{t+1}_c:=\boldsymbol{\beta}_c-\boldsymbol{r}^t_c.
    \end{aligned}
\end{equation}
SC-GAMP decoder is a special case of the general recursion
\begin{equation}
    \begin{aligned}
		&[\boldsymbol{b}^t_r, \boldsymbol{z}_{0r}]-\frac{\sigma^t_r}{\phi^{t-1}_r}[\breve{\boldsymbol{m}}^{t-1}_r,0]=\sum_{c\in[\Gamma]}\sqrt{W_{rc}} \mathbf{A}_{rc}[\breve{\boldsymbol{q}}^t_c, \boldsymbol{\beta}_c],  \\
		&[\boldsymbol{h}^{t+1}_c, -\boldsymbol{\beta}_c] +[\breve{\boldsymbol{q}}^t_c, \boldsymbol{\beta}_c]=\sum_{r\in[\Gamma]}S_{rc}^t\sqrt{W_{rc}}\mathbf{A}_{rc}^T[\breve{\boldsymbol{m}}^t_r,0], 
    \end{aligned}
\end{equation}
where $\boldsymbol{z}_0 = \boldsymbol{A} \boldsymbol{\beta}$,$S_{rc}^t:=\frac{\tau_c^t}{\phi_r^t}$ and $\mathbf{A}_{rc} = \frac{\boldsymbol{A}_{rc}}{\sqrt{W_{rc}}}$ is the modified matrix with entries $\mathrm{A}_{i j} \stackrel{\text { i.i.d. }}{\sim} \mathcal{N}\left(0, \frac{1}{L}\right) $. The reason for introducing $[\boldsymbol{b}^t_r, \boldsymbol{z}_{0r}]$, $[\breve{\boldsymbol{m}}^{t-1}_r,0]$, $[\breve{\boldsymbol{q}}^t_c, \boldsymbol{\beta}_c]$ and $[\boldsymbol{h}_{c}^{t+1}, - \boldsymbol{\beta}_c]$ is to analyze the asymptotic probability distribution of $\boldsymbol{b}_{r}^t$ and $\boldsymbol{z}_0r$. Next we define
\begin{equation}
    \boldsymbol{m}^{t,c}:=\left[\begin{matrix}
    S_{1c}^t\sqrt{W_{1c}}\breve{\boldsymbol{m}}^t_1\\
    S_{2c}^t\sqrt{W_{2c}}\breve{\boldsymbol{m}}^t_2\\
    \vdots\\
    S_{\Gamma c}^t\sqrt{W_{\Gamma c}}\breve{\boldsymbol{m}}^t_\Gamma\\
    \end{matrix}\right],\quad
    \boldsymbol{q}^{t,r}:=\left[\begin{matrix}
    \sqrt{W_{r,4\omega}}\breve{\boldsymbol{q}}^t_{4\omega}\\
    \sqrt{W_{r, 4\omega+1}}\breve{\boldsymbol{q}}^t_{4 \omega + 1}\\
    \vdots\\
    \sqrt{W_{r,\Gamma-4\omega}}\breve{\boldsymbol{q}}^t_{\Gamma - 4\omega}\\
    \end{matrix}\right]
\end{equation}
such that the general recursion is simplified to
\begin{equation}
    \begin{aligned}
		&[\boldsymbol{b}^t_r,\boldsymbol{z}_{0r}]-\frac{\sigma^t_r}{\phi^{t-1}_r}[\breve{\boldsymbol{m}}^{t-1}_r,0]=\sum_{c\in[\hat{\Gamma}]}\mathbf{A}_{rc}[\boldsymbol{q}^{t,r}_c, \sqrt{W_{rc}} \boldsymbol{\beta}_c], \\
		&[\boldsymbol{h}^{t+1}_c, -\boldsymbol{\beta}_c]+[\breve{\boldsymbol{q}}^t_c, \boldsymbol{\beta}_c]=\sum_{r\in[\Gamma]}\mathbf{A}_{rc}^T [\boldsymbol{m}^{t,c}_r, 0].
    \end{aligned}
    \label{eq:gen recursion}
\end{equation}
\eqref{eq:gen recursion} is very similar to the general recursion in \cite{rush2021capacity} with the only difference in the range of $r$ and $c$, so the following conditional distribution lemma (Lemma \ref{lemma:distribution}) is  borrowed from \cite[Lemmas 7.4 and 7.5]{rush2021capacity} with small modification.
	
Before stating the lemma, we have to define several auxiliary matrices. Let
\begin{equation}
    \begin{aligned}
		&\boldsymbol{M}_t^c=[\boldsymbol{m}^{0,c}|\cdots|\boldsymbol{m}^{t-1,c}]\in\mathbb{R}^{n\times t},\\ &\boldsymbol{Q}_t^r=[\boldsymbol{q}^{0,r}|\cdots|\boldsymbol{q}^{t-1,r}]\in\mathbb{R}^{N\times t},
    \end{aligned}
\end{equation}
and
\begin{equation}
    \begin{aligned}
		&\boldsymbol{X}_{t}=[\boldsymbol{h}^1+\breve{\boldsymbol{q}}^0|\cdots|\boldsymbol{h}^t+\breve{\boldsymbol{q}}^{t-1}]\in\mathbb{R}^{N\times t},\\
		&\boldsymbol{Y}_{t}=[\boldsymbol{b}^0|\boldsymbol{b}^1-\boldsymbol{v}^1\odot\breve{\boldsymbol{m}}^0|\cdots|\boldsymbol{b}^{t-1}-\boldsymbol{v}^{t-1}\odot\breve{\boldsymbol{m}}^{t-2}]\in\mathbb{R}^{n\times t},\\
    \end{aligned}
\end{equation}
where $\odot$ denote entry-wise product and $\boldsymbol{v}^t$ is composed of $v_i^t:=\frac{\sigma_r^t}{\phi_r^{t-1}}$ if the $i$-th entry is in the $r$-th block.
	
We further define the orthogonal projection of $\boldsymbol{m}^{t,c}$ and $\boldsymbol{q}^{t,c}$ onto the column space of $\boldsymbol{M}_t^c$ and $\boldsymbol{Q}_t^c$ respectively as
\begin{equation}
    \boldsymbol{m}_\parallel^{t,c}:=\sum_{i=0}^{t-1}\alpha_i^{t,c}\boldsymbol{m}^{i,c},\ \boldsymbol{q}_\parallel^{t,c}:=\sum_{i=0}^{t-1}\gamma_i^{t,c}\boldsymbol{q}^{i,c}
\end{equation}
where $(\alpha_0^{t,c},\cdots,\alpha_{t-1}^{t,c})^T=\boldsymbol{P}^\parallel_{\boldsymbol{M}_t^c}\boldsymbol{m}^{t,c}$ and $(\gamma_0^{t,c},\cdots,\gamma_{t-1}^{t,c})^T=\boldsymbol{P}^\parallel_{\boldsymbol{Q}_t^r}\boldsymbol{q}^{t,c}$ are  coefficient vectors of these projections,  $\boldsymbol{P}^\parallel_{\boldsymbol{M}_t^c}=\boldsymbol{M}_t^c((\boldsymbol{M}_t^c)^T\boldsymbol{M}_t^c)^{-1}(\boldsymbol{M}_t^c)^T$ and $\boldsymbol{P}^\parallel_{\boldsymbol{Q}_t^r}=\boldsymbol{Q}_t^r((\boldsymbol{Q}_t^r)^T\boldsymbol{Q}_t^r)^{-1}(\boldsymbol{Q}_t^r)^T$ are corresponding projection matrices.
	
Let $\sigma^0_{\perp,r}:=\sigma_r^0$ and $\tau^0_{\perp,r}:=\tau_r^0$ and for $t\geq1$ define
\begin{equation}
    \sigma^t_{\perp,r}:=\sigma_r^t(1-\frac{\sigma_r^t}{\sigma_r^{t-1}}),\ \tau^t_{\perp,c}:=\tau_c^t(1-\frac{\tau_r^t}{\tau_c^{t-1}}).
\end{equation}
Lastly, to specify the conditional distribution, we use $\mathscr{F}_{t_a,t}$ to denote the sigma algebra generated by the collection of vectors
\begin{equation*}
    \boldsymbol{b}^0,\cdots,\boldsymbol{b}^{t_a-1},\breve{\boldsymbol{m}}^0,\cdots,\breve{\boldsymbol{m}}^{t_a-1},\boldsymbol{h}^1,\cdots,\boldsymbol{h}^t,\breve{\boldsymbol{q}^0,\cdots,\breve{\boldsymbol{q}}^t}\ and\ \boldsymbol{\beta}.
\end{equation*}
\begin{lemma}
\label{le: bound}
    Under Assumption \ref{assumption3}, for sufficiently large $M$, the constants $\sigma^t_{\perp,r}$ and $\frac{n}{L} \tau^{t}_{\perp,c}$ are bounded below for $0 \leq k < T$:
    \begin{equation}
        \sigma_{\perp,t}^t \geq \overline{C}_1 ^2(\frac{\omega}{ \Gamma})^2, \quad \frac{n}{L} \tau_{\perp,c}^t \geq \overline{C}_2 \frac{\omega}{\Gamma}
    \end{equation}
    where
    \begin{equation}
        \overline{C}_1 = C^2, \overline{C}_2 = \frac{LC}{2 f_{\text{out}}^2 (C\frac{\omega}{\Gamma})}
    \end{equation}
    with the constant $C = \frac{2 \rho g}{\omega}$ where $g$ is defined in (\ref{eq:g(Delta)}).
\end{lemma}
This lemma indicates that both $\sigma_{\perp,r}^t$ and $\frac{n}{L} \tau^t_{\perp,c}$ are bounded throughout the iteration, thereby enabling to treat them as $O(1)$ quantities independent of $n,L,M$. Additionally, it confirms the justification of assumption made in Remark $2$.
\begin{lemma}
    For the general recursion in \eqref{eq:gen recursion}, we have
    \begin{equation}
		\begin{aligned}
		  \boldsymbol{b}^t_r|_{\mathscr{S}_{t,t}}\overset{d}{=}\sum_{i=0}^t\frac{\sigma_r^t}{\sigma_r^i}(\sqrt{\sigma^i_{\perp,r}}\boldsymbol{Z}_{i,r}'+\boldsymbol{\Delta}_{i,i,r}),\ r\in[\Gamma],\\
		  \boldsymbol{h}^{t+1}_c|_{\mathscr{S}_{t+1,t}}\overset{d}{=}\sum_{i=0}^t\frac{\tau_c^t}{\tau_c^i}(\sqrt{\tau^i_{\perp,c}}\boldsymbol{Z}_{i,c}+\boldsymbol{\Delta}_{i+1,i,r}),\ c\in[\hat{\Gamma}] ,
		\end{aligned}
    \end{equation}
    where $\boldsymbol{Z}_{i,r}'\sim\mathcal{N}(0,\boldsymbol{I}_n)$ is independent from $\mathscr{S}_{t,t}$ and $\boldsymbol{Z}_{i,c}\sim\mathcal{N}(0,\boldsymbol{I}_n)$ is independent from $\mathscr{S}_{t+1,t}$. Deviation vectors are given by $\boldsymbol{\Delta}_{0,0,r} = \boldsymbol{0}$ and 
    \begin{equation}
		\begin{aligned}
		  \boldsymbol{\Delta}_{1,0,c}=&\left[\frac{1}{\sqrt{L}}||\boldsymbol{m}^{0,c}||-\sqrt{\tau}_c^0\right]\boldsymbol{Z}_{0,c}-\sum_{r\in[\Gamma]}\frac{1}{\sqrt{L}}||\boldsymbol{m}_r^{0,c}||[\boldsymbol{P}_{\boldsymbol{Q}_1^t}^\parallel\boldsymbol{Z}_0^r]_c\\&+\breve{\boldsymbol{q}}_c^0\left(\sum_{r\in[\Gamma]}\frac{\sqrt{W_{rc}(\boldsymbol{b}_r^0)^*\boldsymbol{m}_r^{0,c}}}{L\sigma_r^0}-1\right)
	   \end{aligned}
    \end{equation}
    and for $t>0$,
    \begin{equation}
		\begin{aligned}                   
            &\boldsymbol{\Delta}_{t,t,r}=\sum_{i=0}^{t+2}\boldsymbol{b}_r^i\gamma_i^{t,r}+\boldsymbol{b}_r^{t-1}\left[\gamma_{t-1}^{t,r}-\frac{\sigma_r^t}{\sigma_r^{t-1}}\right]\\
		  &+\left[\frac{1}{\sqrt{L}}||\boldsymbol{q}_\perp^{t,r}-\sqrt{\sigma_{\perp,r}^t}||\right]\boldsymbol{Z}_{t,r}'-\sum_{c\in[\hat{\Gamma}]}\frac{1}{\sqrt{L}}\left|\left|\boldsymbol{q}^{t,r}_{\perp,c}[\boldsymbol{P}^\parallel_{\boldsymbol{M}_t^c\boldsymbol{Z}_t^{'c}}]_r\right|\right|\\
            &+\sum_{c\in[\hat{\Gamma}]}\boldsymbol{M}_{t,r}^c((\boldsymbol{M}_t^c)^T\boldsymbol{M}_t^c)^{-1}(\boldsymbol{X}_{t,c})^T\boldsymbol{q}_{\perp,c}^{t,r}\\
	       &-\sum_{i=1}^{t-1}\gamma_i^{t,r}v_r^i\breve{\boldsymbol{m}}_r^{i-1}+v_r^t\breve{\boldsymbol{m}}_r^{t-1},
		\end{aligned}
    \end{equation}
    \begin{equation}
		\begin{aligned}
		    &\boldsymbol{\Delta}_{t+1,t,c}=\sum_{i=0}^{t+2}\boldsymbol{h}_c^{i+1}\alpha_i^{t,r}+\boldsymbol{h}_c^t\left[\alpha_{t-1}^{t,c}-\frac{\tau_c^t}{\tau_c^{t-1}}\right]\\
		  &+\left[\frac{1}{\sqrt{L}}||\boldsymbol{m}_\perp^{t,r}-\sqrt{\tau_{\perp,r}^t}||\right]\boldsymbol{Z}_{t,c}-\sum_{r\in[\Gamma]}\frac{1}{\sqrt{L}}\left|\left|\boldsymbol{m}^{t,c}_{\perp,r}[\boldsymbol{P}^\parallel_{\boldsymbol{Q}_{t+1}^r\boldsymbol{Z}_t^{r}}]_c\right|\right|\\
		  &+\sum_{r\in[\Gamma]}\boldsymbol{Q}_{t+1,c}^r((\boldsymbol{Q}_{t+1}^r)^T\boldsymbol{Q}_{t+1}^r)^{-1}(\boldsymbol{Y}_{t+1,r})^T\boldsymbol{m}_{\perp,r}^{t,c}\\&-\sum_{i=0}^{t-1}\alpha_i^{t,c}\breve{\boldsymbol{q}}_c^i-\breve{\boldsymbol{q}}_r^t.
		\end{aligned}
    \end{equation}
\label{lemma:distribution}
\end{lemma}

This lemma demonstrates that both $\boldsymbol{b}^t_r$ and $\boldsymbol{h}_c^{t+1}$ asymptotically converge to a Gaussian distribution plus a deviation term. Next we will show that the deviation terms concentrates on zero exponentially fast. Note that it is different from \cite[Lemma 7.5]{rush2021capacity} due to the introduction of the general channel and the corresponding denoiser $\boldsymbol{g}_{\text{out}}$.

\begin{lemma}
Let $\epsilon \in (0,1)$ and $0 \leq t < T$, where $T$ is defined in (\ref{eq:T}). Iteration-dependent quantities summarizing the problem parameters are defined as 
    \begin{equation}
        \begin{aligned}
            \Pi^t = \Gamma^{2t+1}, \Pi_t^{\prime} = \Gamma^{2t+2} \\
            \pi_t =\pi_t^{\prime} = \frac{n(\omega / \Gamma)^{2t+3}}{(\log M)^{2t+2}}
        \end{aligned}
    \end{equation}

\noindent
\textbf{(a)} Let $u$ be an integer with $u \in \{0,1,2\}$ and $(x)_{+}$ is defined as $\max \{x,0\}$. For all $0\leq s\leq t$
\begin{equation*}
    \begin{aligned}
        & \mathbb{P}\left(\frac{1}{n} \sum_{r \in[\Gamma]} W_{r c}^{\mathrm{u}}\left\|\boldsymbol{\Delta}_{t, t, r}\right\|^{2} \geq \epsilon\right) \\
        &\qquad \leq t^{3} K K_{t-1} \Pi_{t-1} \exp \left\{\frac{-\kappa \kappa_{t-1}(\omega / \Gamma)^{(\mathrm{u}-1)_{+}} \pi_{t-1} \epsilon}{t^{6}}\right\}
    \end{aligned} 
\end{equation*} 
\begin{equation*}
    \begin{aligned}
    & \mathbb{P} (|\frac{1}{n} \sum_{r \in[\Gamma]} W_{r c}^{\mathrm{u}}\left(\boldsymbol{b}_{r}^{s}\right)^{*} \boldsymbol{b}_{r}^{t} - \frac{1}{R} \sum_{r \in[\Gamma]} W_{r c}^{\mathrm{u}} \sigma_{r}^{t}| \geq \epsilon)  \\
    & \qquad \leq t^{4} K K_{t-1} \Pi_{t-1} \exp \left\{\frac{-\kappa \kappa_{t-1}(\omega / \Gamma)^{2(\mathrm{u}-1)_{+}} \pi_{t-1} \epsilon^{2}}{t^{8}}\right\} .
    \end{aligned}
\end{equation*}
\textbf{(b)} Let $v$ be an integer with $v \in \{0,1\}$, we have the following for all $0\leq s\leq t+1$:
\begin{equation}
    \begin{aligned}
        &  \mathbb{P} \left(\frac{1}{L} \sum_{c \in[\hat{\Gamma}]} W_{r c}^{2 v} \sum_{\ell \in c} \max _{j \in \sec (\ell)}\left|\left[\boldsymbol{\Delta}_{t+1, t, c]}\right]_j\right|^2 \geq \epsilon\right) \\
        & \qquad \leq t^3 K K_{t-1}^{\prime} \Pi_{t-1}^{\prime} \exp \left\{\frac{-\kappa \kappa_{t-1}^{\prime}(\omega / \Gamma)^{2 v} \pi_{t-1}^{\prime} \epsilon}{t^6}\right\}
    \end{aligned}
\end{equation}
\begin{equation} \label{ineq:q}
    \begin{aligned}
        & \mathbb{P} \left(\left|\sum_{c \in[\hat{\Gamma}]} W_{r c}^v\left[\frac{1}{L}\left(\breve{\boldsymbol{q}}_c^s\right)^* \breve{\boldsymbol{q}}_c^{t+1}-\frac{1}{C} \psi_c^{t+1}\right]\right| \geq \epsilon\right)  \\ 
        & \qquad \leq t^4 K K_{t-1}^{\prime} \Pi_{t-1}^{\prime} \exp \left\{\frac{-\kappa \kappa_{t-1}^{\prime}(\omega / \Gamma)^{2 v} \pi_{t-1}^{\prime} \epsilon^2}{t^8(\log M)^2}\right\}
    \end{aligned}
\end{equation}
\label{lemma:concentration}
\end{lemma}

Lemma \ref{lemma:concentration} is the main concentration lemma. The primary technique is the conditional technique as shown in \cite{bayati2011dynamics, rush2021capacity} and we should tackle the convergence of terms related to estimator $\textbf{g}_{out}$ carefully. It is the counterpart of \cite[Lemma 7.6]{rush2021capacity} and will be proved in the longer version. Theorem \ref{theo:main} can be easily proved through the combination of Lemma \ref{lemma:SE} and Lemma \ref{lemma:concentration}.

\begin{proof}[Proof of Theorem \ref{theo:main}]
    Recall from (\ref{eq:general}) that $\breve{\boldsymbol{q}}^t_c=\boldsymbol{\beta}^t_c-\boldsymbol{\beta}_c$. By setting $v=0$, (\ref{ineq:q}) implies for $0 \leq t \leq T(-1)$,
    \begin{equation}
        \begin{aligned}
             & P\left(\left|\frac{\left\|\boldsymbol{\beta}^{t+1}-\boldsymbol{\beta}\right\|^{2}}{L}-\frac{1}{\Gamma} \sum_{c \in[\Gamma]} \psi_{c}^{t}\right| \geq \epsilon\right)  \\
             & \qquad \qquad \qquad  \leq K_t \Gamma^{2t+3} \exp \{ \frac{-k_t N \omega^{2t+1} \epsilon^2}{(\log M)^{2t+2}\Gamma^{2t+1}}\}.
        \end{aligned}
    \end{equation}
Regarding the SPARC message vector, the section error rate can be bounded in terms of the MSE $\| \boldsymbol{\beta}^T_{\ell} - \boldsymbol{\beta}_{{\ell}} \|_2^2$, indicating
\begin{equation}
    \boldsymbol{\beta}^T_{\ell} \neq \boldsymbol{\beta}_{{\ell}} \quad \Rightarrow \quad \| \boldsymbol{\beta}^T_{\ell} - \boldsymbol{\beta}_{{\ell}} \|_2^2 \geq \frac{1}{4},
\end{equation}
By substituting $\epsilon$ with $\frac{\epsilon}{4}$, we obtain the concentration inequality for the section error rate in (\ref{eq:main}).
\end{proof}

\subsection{Performance of SC-GAMP algorithm for GLM} \label{app:glm}

In this section, we give a non-asymptotic analysis for SC-GAMP algorithm for GLM. The main difference in GLM lies in the assumption that the entries of the signal vector $\boldsymbol{\beta}$ are drawn from a generic prior $P_{\beta}$, as opposed to the section-wise structure inherent in a SPARC message vector. We adjust the variances of entries in $\mathbf{A}_{rc}$ to $\frac{W{rc}}{N/\Gamma}$. The estimator $\mathbf{g}_{\text{in}}(\boldsymbol{r}, \tau)$ now is:
\begin{equation}
    \mathbf{g}_{\text{in}}(\boldsymbol{r}, \tau) = \mathbb{E}[\mathbf{R}|\mathbf{R}+ \sqrt{\tau} \mathbf{Z} = \boldsymbol{r}],
\end{equation}
where $\mathbf{R} \sim P_{\beta}^{\otimes N/\Gamma}$ and $\mathbf{Z} \sim \mathcal{N}(0, \mathbf{I}_{N/ \Gamma})$. Then $\psi_{c}^{t+1}$ related to $\mathbf{g}_{\text{in}}(\boldsymbol{r}, \tau)$ is changed to:
\begin{equation}
    \psi_c^{t+1} = \mathbb{E} {[\beta - g_{\text{in}}(\sqrt{\tau} G + \beta, \tau)]^2},
\end{equation}
where the expectation is taken over $\beta \sim P_\beta$, $G \sim \mathcal{N}(0,1)$ independent of $\beta$. We make following assumptions for GLM and $g_{\text{in}}(r, \tau)$.
\begin{assumption} \label{assumption4}
    The components of the signal vector $\boldsymbol{\beta} \in \mathbb{R}^{N}$ are i.i.d sampled from a sub-Gaussian distribution $P_{\beta}$.
\end{assumption}
\begin{assumption} \label{assumption5}
    With growth of the signal dimension $N$, the sampling ratio $n/N$ remains constant and is denoted by $\alpha$.
\end{assumption}
\begin{assumption} \label{assumption6}
    The estimator $g_{\text{in}}(r, \tau)$ is continuous differentiable and Lipschitz with respect to the variable $r$ for different $\tau \in [0, \mathbb{E} \beta^2]$.
\end{assumption}

\noindent
\textbf{Remark 3}. It is easy to check that Gaussian, Bernoulli and Gaussian-Bernoulli distributions all satisfy Assumption \ref{assumption4} and \ref{assumption6}.

We present the non-asymptotic convergence result for the MSE of the SG-GAMP algorithm as follows:
\begin{theorem}
\label{theo:theo2}
Under Assumptions $\ref{assumption1}$ to $\ref{assumption6}$, for $t\geq 1$, the MSE of the SG-GAMP algorithm satisfies
\begin{equation*}
    P\left(\left|\frac{\left\|\boldsymbol{\beta}^{t}-\boldsymbol{\beta}\right\|^{2}}{N}-\frac{1}{\Gamma} \sum_{c \in[\Gamma]} \psi_{c}^{t}\right| \geq \epsilon\right) \leq K_t \Gamma^{2t+1} \exp \{ \frac{-k_t n \omega^{2t} \epsilon^2}{\Gamma^{2t+1}}\}.
\end{equation*}
Here $K_t$, $k_t$ are positive constants independent of $\epsilon, \Gamma, \omega$. (These constants differ from the ones in Theorem \ref{theo:main}).
\end{theorem}
Theorem \ref{theo:theo2} can be seen as a extended non-asymptotic analysis of the concentration of SE in \cite{donoho2013information}.

\end{document}